\documentclass[12pt]{article}
\usepackage[english]{babel}
\usepackage{amsmath}
\usepackage{amsthm}
\usepackage{amsfonts}
\usepackage{amssymb}
\usepackage{authblk}
\usepackage{fullpage}
\usepackage{cite}
\usepackage{comment}
\usepackage{mathtools}
\usepackage[colorlinks=true]{hyperref}
\newtheorem{thm}{Theorem}[section]

\newtheorem*{thm*}{Theorem}
\newtheorem{cor}{Corollary}[section]

\newtheorem{lem}{Lemma}[section]

\newtheorem{prop}{Proposition}[section]

\theoremstyle{definition}
\newtheorem{defn}{Definition}[section]

\theoremstyle{remark}
\newtheorem{rem}{Remark}[section]

\numberwithin{equation}{section}
\newcommand{\ket}[1]{|{#1}\rangle}
\newcommand{\bra}[1]{\langle{#1}|}
\def\bI{\mathbb I}
\def\mU{\mathcal{U}}
\def\mM{\mathcal{M}}
\def\beq{\begin{equation}}
\def\eeq{\end{equation}}
\def\bC{\mathbb C}
\def\bE{\mathbb E}
\def\tr{\mbox{Tr}}
\DeclarePairedDelimiter{\ceil}{\lceil}{\rceil}

\title{Classical shadows meet quantum optimal mass transport}
\author[1]{Giacomo De Palma}
\author[1,2]{Tristan Klein}
\author[1,3]{Davide Pastorello}
\affil[1]{University of Bologna, Department of Mathematics, Piazza di Porta San Donato 5, 40126 Bologna, Italy}
\affil[2]{ENS de Lyon, Département Informatique, 15 parvis René Descartes 69342 Lyon Cedex 07, France}
\affil[3]{TIFPA-INFN, via Sommarive 14, 38123 Povo (Trento), Italy}

\begin{document}

\maketitle

\begin{abstract}
  \noindent  Classical shadows constitute a protocol to estimate the expectation values of a collection of $M$ observables acting on $O(1)$ qubits of an unknown $n$-qubit state with a number of measurements that is independent of $n$ and that grows only logarithmically with $M$. We propose a local variant of the quantum Wasserstein distance of order 1 of [De Palma \emph{et al.}, IEEE Trans. Inf. Theory 67, 6627 (2021)] and prove that the classical shadow obtained measuring $O(\log n)$ copies of the state to be learned constitutes an accurate estimate with respect to the proposed distance.
  We apply the results to quantum generative adversarial networks, showing that quantum access to the state to be learned can be useful only when some prior information on such state is available.       

\end{abstract}

\section{Introduction}
Quantum tomography consists in finding a classical estimate of an unknown quantum state by measuring a given number of independent copies of the state and plays a key role in quantum information science \cite{ariano2003quantum}.
The number of parameters that are required to describe a generic state of $n$ qubits grows exponentially with $n$.
Therefore, if no information on the state is known a priori, any estimate that is accurate with respect to the trace distance or the fidelity requires an exponential number of copies of the state and becomes quickly unfeasible even for moderately large $n$ \cite{donnell2016efficient,haah2017sample,anshu2023survey}.
However, the situation becomes radically different if we weaken the metric employed to measure the quality of the estimate.
Let us consider the scenario where we are only interested in the expectation values of some observables.
Then, the shadow tomography protocol can estimate such expectation values with a number of copies of the state that scales linearly with the number of qubits and polylogarithmically with the number of observables \cite{aaronson2020shadow}.
This result can be further improved if the observables are tensor products of $O(1)$ Pauli matrices.
Indeed in this case, by measuring a number of copies of the state that scales logarithmically with the number of observables and that is independent on the number of qubits, the classical shadow protocol can generate a classical estimate of the state (the classical shadow) from which the expectation values of the observables can be estimated \cite{huang2020predicting,huang2022learning}.
The striking property of this protocol is that the classical shadow does not depend on the observables to be estimated, and the protocol works also if such observables are revealed only after the measurements have been performed.
A further improvement of the classical shadow protocol allows to estimate the expectation value of the tensor products of any number of Pauli matrices \cite{Huang_2021}.

In this paper, we propose a distance on the set of the states of $n$ qubits that metrizes the convergence of the classical shadow to the state to be estimated, \emph{i.e.}, such that the classical shadow obtained by measuring $O\left(\frac{1}{\epsilon^2}\log\frac{1}{\epsilon}\log n\right)$ copies of the state to be estimated achieves with high probability distance $\epsilon$ from the state.
This distance, which we call the local quantum $W_1$ distance, is built upon the quantum theory of optimal mass transport and is a variant of the quantum Wasserstein distance of order $1$ (or quantum $W_1$ distance) for qubits proposed in \cite{de2021quantum} (several other approaches to quantum optimal mass transport have been proposed, the most relevant are summarized in \autoref{app:appr}).
The local quantum $W_1$ distance is induced by the norm dual to the local quantum norm, \emph{i.e.}, the local quantum $W_1$ distance between the states $\rho$ and $\sigma$ is the maximum of the difference between the expectation values on $\rho$ and on $\sigma$ of an observable with local quantum norm at most one.
The local quantum norm is inspired to the norm employed in the study of quantum spin systems on infinite lattices to turn the space of the local interactions into a Banach space \cite{bratteli2013operatorII,naaijkens2017quantum}, and is built to have low values on observables which are the sum of operators acting on few qubits.

As the quantum $W_1$ distance of \cite{de2021quantum}, the local quantum $W_1$ distance coincides with the trace distance for $n=1$ and is an extensive quantity, \emph{i.e.}, it is superadditive with respect to the composition of quantum systems and additive for product states.
In particular, the local quantum $W_1$ distance recovers the Hamming distance for the states of the computational basis.
The local quantum $W_1$ distance is always upper bounded by the quantum $W_1$ distance.
However, the local norm imposes a much stronger constraint on the observables than the quantum Lipschitz norm dual to the quantum $W_1$ distance, and the local quantum $W_1$ distance between states that are locally indistinguishable can be exponentially smaller than the quantum $W_1$ distance.

Our results are complementary to the results of \cite{rouze2023learning,onorati2023efficient}, which provide a protocol to estimate with respect to the quantum $W_1$ distance of \cite{de2021quantum} any quantum state satisfying a transportation-cost inequality with a number of copies that grows polylogarithmically with the number of qubits.
Quantum transportation-cost inequalities provide an upper bound to the quantum $W_1$ distance in terms of the quantum relative entropy, and have been proved for Gibbs states of local Hamiltonians with a sufficiently strong decay of correlations \cite{de2022quantum,onorati2023efficient}.
This paper removes any assumption on the state to be learned by weakening the metric employed to measure the quality of the estimate.
However, we prove that when restricted to a set of Gibbs states of local Hamiltonians that can be learned efficiently in the quantum $W_1$ distance, the local quantum $W_1$ distance proposed in this paper is equivalent to the quantum $W_1$ distance.

We apply our results to the Quantum Wasserstein Generative Adversarial Network (QWGAN) proposed in \cite{kiani2022learning}, which provides an algorithm to train a variational quantum circuit to learn an unknown quantum state.
We show that, if no a priori information on the state is available, the QWGAN can be equivalently trained on a classical shadow of the state and does not gain any advantage from having quantum access to the state.

The paper is structured as follows.
In \autoref{sec:LQD} we define the local norm and the local quantum $W_1$ distance and prove some of their properties.
In \autoref{sec:ClassicalShadows}, we present the classical shadow protocol of \cite{huang2020predicting} and its improvement of \cite{Huang_2021}.
In \autoref{sec:conv}, we determine the convergence rate of the classical shadow protocol with respect to the local quantum $W_1$ distance.
In \autoref{sec:Gibbs}, we prove the equivalence between the quantum $W_1$ distance and the local quantum $W_1$ distance for the Gibbs states of local Hamiltonians satisfying a transportation-cost inequality.
In \autoref{sec:GAN}, we discuss the application to QWGANs.
We conclude in \autoref{sec:concl}.
\autoref{app:appr} summarizes the main approaches to quantum optimal mass transport, \autoref{app:shadows} presents some related works on classical shadows, and \autoref{app:lem} contains the proofs of the auxiliary lemmas.

\section{The local quantum \texorpdfstring{$W_1$}{W1} distance}\label{sec:LQD}
In this section, we recall the definition of the quantum Wasserstein distance of order $1$, we introduce the local quantum $W_1$ distance and prove its basic properties.

\subsection{Notation}
Let us start by setting the notation for the paper:
\begin{defn}
    For any $k\in\mathbb{N}$ we define
    \begin{equation}
        [k] = \{1,\,\ldots,\,k\}\,.
    \end{equation}
\end{defn}

We consider a quantum system made by $n$ qubits, which we label with the integers from $1$ to $n$.
Each qubit is associated with the Hilbert space $\mathbb{C}^2$, such that the Hilbert space of the system is $\left(\mathbb{C}^2\right)^{\otimes n}$.
\begin{defn}
    For any subset of the qubits $\Lambda\subseteq[n]$, let
    \begin{equation}
        \mathcal{H}_\Lambda = \bigotimes_{x\in\Lambda}\mathbb{C}^2
    \end{equation}
    be the Hilbert space associated with the qubits in $\Lambda$, let $\mathcal{O}_\Lambda$ be the set of the self-adjoint linear operators acting on $\mathcal{H}_\Lambda$, let $\mathcal{O}_\Lambda^T$ be the set of the traceless operators in $\mathcal{O}_\Lambda$, and let $\mathcal{S}_\Lambda$ be the set of the quantum states acting on $\mathcal{H}_\Lambda$.
\end{defn}

\subsection{The quantum \texorpdfstring{$W_1$}{W1} distance and the quantum Lipschitz constant}
In this subsection, we briefly present the quantum $W_1$ distance and the quantum Lipschitz constant of \cite{de2021quantum}.

The quantum $W_1$ distance is based on the notion of neighboring quantum states.
Two states of $n$ qubits are neighboring if they coincide after discarding a suitable qubit.
We define the quantum $W_1$ norm $\|\cdot\|_{W_1}$ as the maximum norm that assigns distance at most one to any couple of neighboring states.
The quantum $W_1$ distance is then the distance induced by the quantum $W_1$ norm.
More formally, we have the following:
\begin{defn}[Quantum $W_1$ norm]
For any $\Delta\in\mathcal{O}_{[n]}^T$ we define
\begin{equation}
    \|\Delta\|_{W_1} = \min\left\{\frac{1}{2}\sum_{x\in[n]}\left\|\Delta^{(x)}\right\|_1 : \Delta^{(x)}\in\mathcal{O}_{[n]}^T\,,\;\mathrm{Tr}_x\Delta^{(x)} = 0\;\forall\,x\in[n]\,,\;\Delta = \sum_{x\in[n]}\Delta^{(x)}\right\}\,.
\end{equation}
\end{defn}
The quantum $W_1$ distance can be thought as a quantum version of the Hamming distance, since it exactly recovers the Hamming distance for the states of the computational basis.

We define the dependence of the observable $H$ on the qubit $x$ as twice the minimum operator norm of the difference between $H$ and any observable that does not act on $x$:
\begin{defn}[\!\!\cite{de2022wasserstein}]\label{def:partial}
    For any $x\in[n]$ and any $H\in\mathcal{O}_{[n]}$ we define
    \begin{equation}
        \partial_x H = 2\min_{K\in\mathcal{O}_{x^c}}\left\|H - K\right\|_\infty\,,
    \end{equation}
    where $x^c=[n]\setminus \{x\}$
\end{defn}

We then define the quantum Lipschitz constant of the observable $H$ as the maximum dependence of $H$ on a qubit:
\begin{defn}[Quantum Lipschitz constant]
    For any $H\in\mathcal{O}_{[n]}$ we define
    \begin{equation}
        \|H\|_L = \max_{x\in[n]}\partial_x H\,.
    \end{equation}
\end{defn}

The quantum $W_1$ norm on $\mathcal{O}_{[n]}^T$ and the quantum Lipschitz constant on $\mathcal{O}_{[n]}$ are mutually dual:
\begin{prop}[\!\!\cite{de2021quantum}]\label{prop:W1L}
    For any $\Delta\in\mathcal{O}_{[n]}^T$ we have
    \begin{equation}\label{eq:W1L}
    \|\Delta\|_{W_1} = \max_{H\in\mathcal{O}_{[n]} : \|H\|_L\le1}\mathrm{Tr}\left[\Delta\,H\right]\,.
    \end{equation}
\end{prop}

Despite the fact that the Lipschitz constant seems to constrain the maximization in \eqref{eq:W1L} to local observables, the quantum $W_1$ distance between states that are locally indistinguishable can be large.
This is a consequence of the continuity of the von Neumann entropy $S(\rho) = -\mathrm{Tr}\left[\rho\ln\rho\right]$ with respect to the quantum $W_1$ distance:
\begin{thm}[\!\!\cite{de2022wasserstein}]\label{thm:S}
For any two states of $n$ qubits $\rho,\,\sigma\in\mathcal{S}_{[n]}$,
\begin{equation}
\frac{\left|S(\rho) - S(\sigma)\right|}{n} \le h_2\left(\frac{\left\|\rho-\sigma\right\|_{W_1}}{n}\right) + \frac{\left\|\rho-\sigma\right\|_{W_1}}{n}\ln3\,,
\end{equation}
where $h_2$ is the binary entropy function
\begin{equation}
h_2(x) = -x\ln x - \left(1-x\right)\ln\left(1-x\right)\,,\qquad 0\le x\le 1\,.
\end{equation}
\end{thm}
Indeed, on the one hand \autoref{thm:S} implies that any pure state is far from the maximally mixed state:
\begin{prop}
For any $n\in\mathbb{N}$ and any pure state of $n$ qubits $\rho\in\mathcal{S}_{[n]}$ we have
\begin{equation}
\left\|\rho - \frac{\mathbb{I}}{2^n}\right\|_{W_1} > 0.189\,n\,.
\end{equation}
\end{prop}
\begin{proof}
Setting
\begin{equation}
    \left\|\rho - \frac{\mathbb{I}}{2^n}\right\|_{W_1} = n\,w\,,\qquad 0\le w\le 1\,,
\end{equation}
\autoref{thm:S} implies
\begin{equation}
h_2(w) + w\ln 3 \ge \ln 2\,,
\end{equation}
from which $w>0.189$.
\end{proof}

On the other hand, pure states that are locally indistinguishable from the maximally mixed state do exist:
\begin{prop}[\!\!\cite{arnaud2013exploring}]
For any $n$ sufficiently large, there exists a pure state of $n$ qubits $\rho\in\mathcal{S}_{[n]}$ such that for any region $\Lambda\subset[n]$ of size $|\Lambda|\le\lfloor0.189\,n\rfloor$, the marginal of $\rho$ on $\Lambda$ is maximally mixed.
\end{prop}

The local quantum $W_1$ distance will capture the property of local distinguishability, and will take an exponentially small value for any two locally indistinguishable states.

\subsection{The local quantum norm and the local quantum \texorpdfstring{$W_1$}{W1} distance}
In analogy to \autoref{prop:W1L}, we wish to define the local quantum $W_1$ norm as the dual of a local quantum norm for observables.
Since we want the local quantum $W_1$ distance to capture the property of local distinguishability, we require the local quantum $W_1$ norm of a sum of operators acting on few qubits to be small.

We consider all the decompositions of an observable $H$ as a sum of local operators.
In analogy to \autoref{def:partial}, we define the dependence of any such decomposition on a qubit $x$ as the sum of the operator norm of each local operator that acts on $x$ weighted by a penalty $c_k$ that grows with the locality $k$ of the operator (\emph{i.e.}, $k$ is the number of qubits on which the operator acts).
We then define the local norm of such decomposition as the maximum dependence on a qubit, and the local norm of $H$ as the minimum local norm of all its possible decompositions:
\begin{defn}[Local quantum norm]\label{local norm}
Let $1 = c_1 \le \ldots \le c_n\,.$
For any $H\in\mathcal{O}_{[n]}$ we define
\begin{equation}\label{eq:loc}
\|H\|_{\mathrm{loc}} = 2\min\left\{\max_{x\in[n]}\sum_{\Lambda\ni x}c_{\left|\Lambda\right|}\left\|H_\Lambda\right\|_\infty : H = \sum_{\Lambda\subseteq[n]}H_\Lambda\,,\; H_\Lambda\in\mathcal{O}_\Lambda\right\}\,.
\end{equation}
\end{defn}

\begin{rem}
The local quantum norm depends on the choice of the penalties $\{c_k\}_{k\in\mathbb{N}}$.
Such norm is analog to the norm defined in \cite{bratteli2013operatorII,naaijkens2017quantum} in the context of quantum spin systems on infinite lattices to turn the set of interactions into a Banach space.
\cite{bratteli2013operatorII,naaijkens2017quantum} define an interaction through its decomposition as a sum of local operators, so their norm does not involve the minimization over the decompositions.
\cite{bratteli2013operatorII,naaijkens2017quantum} choose the penalties to grow exponentially with the size and eventually with the diameter of the region over which the operator acts, but at this stage we prefer to keep the freedom in the choice of the penalties.
\end{rem}

We can now define the local quantum $W_1$ norm as the dual of the local quantum norm:
\begin{defn}[Local quantum $W_1$ norm]
    We define the local quantum $W_1$ norm as the norm on $\mathcal{O}_{[n]}^T$ that is dual to the local norm on $\mathcal{O}_{[n]}$: For any $\Delta\in\mathcal{O}_{[n]}^T$,
    \begin{equation}\label{eq:defW1loc}
        \|\Delta\|_{W_1\mathrm{loc}} = \max\left\{\mathrm{Tr}\left[\Delta\,H\right] : H\in\mathcal{O}_{[n]}\,,\|H\|_{\mathrm{loc}}\le1\right\}\,.
    \end{equation}
\end{defn}

The local quantum $W_1$ norm can be computed with a linear program.
\eqref{eq:defW1loc} constitutes the dual program, while the primal program is provided by the following:
\begin{prop}
    For any $\Delta\in\mathcal{O}_{[n]}^T$ we have
    \begin{equation}\label{eq:Wloc}
        \|\Delta\|_{W_1\mathrm{loc}} = \min\left\{\sum_{x\in[n]}a_x : \frac{\left\|\mathrm{Tr}_{\Lambda^c}\Delta\right\|_1}{2\,c_{|\Lambda|}} \le \sum_{x\in\Lambda}a_x\;\forall\,\Lambda\subseteq[n]\right\}\,.
    \end{equation}
\end{prop}

\begin{proof}
    We have
    \begin{align}\label{eq:Wloclp}
        \|\Delta\|_{W_1\mathrm{loc}} &= \max\left\{\sum_{\Lambda\subseteq[n]}\mathrm{Tr}\left[\Delta\,H_\Lambda\right] : H_\Lambda\in\mathcal{O}_\Lambda\,,\;2\sum_{\Lambda\ni x}c_{|\Lambda|}\left\|H_\Lambda\right\|_\infty \le 1\;\forall\,x\in[n]\right\}\nonumber\\
        &= \max\left\{\sum_{\Lambda\subseteq[n]}\left\|\mathrm{Tr}_{\Lambda^c}\Delta\right\|_1\left\|H_\Lambda\right\|_\infty : H_\Lambda\in\mathcal{O}_\Lambda\,,\;2\sum_{\Lambda\ni x}c_{|\Lambda|}\left\|H_\Lambda\right\|_\infty \le 1\;\forall\,x\in[n]\right\}\nonumber\\
        &= \max\left\{\sum_{\Lambda\subseteq[n]}\left\|\mathrm{Tr}_{\Lambda^c}\Delta\right\|_1 t_\Lambda : t_\Lambda\ge0\;\forall\,\Lambda\subseteq[n]\,,\;2\sum_{\Lambda\ni x}c_{|\Lambda|}\,t_\Lambda \le 1\;\forall\,x\in[n]\right\}\,.
    \end{align}
    The last maximization in \eqref{eq:Wloclp} is the dual program of the primal linear program \eqref{eq:Wloc}.
    The claim follows.
\end{proof}

\begin{rem}
For $n=1$, the local quantum $W_1$ norm coincides with one half times the trace norm.
\end{rem}

\subsection{Properties of the local quantum \texorpdfstring{$W_1$}{W1} distance}
In this subsection we prove some basic properties of the local quantum $W_1$ distance.

\begin{itemize}
\item The local quantum $W_1$ distance always lies between the trace distance divided by the maximum penalty and the quantum $W_1$ distance:
\begin{prop}\label{prop:Lloc}
    We have
    \begin{equation}
        \|\cdot\|_L \le \|\cdot\|_{\mathrm{loc}} \le 2\,c_n\,\|\cdot\|_\infty\,,\qquad \frac{\|\cdot\|_1}{2\,c_n} \le \|\cdot\|_{W_1\mathrm{loc}} \le \|\cdot\|_{W_1}\,.
    \end{equation}
\end{prop}

\begin{proof}
Let $H\in\mathcal{O}_{[n]}$.
Choosing $H_{[n]} = H$ in \eqref{eq:loc} we get
\begin{equation}
    \|H\|_{\mathrm{loc}} \le 2\,c_n\left\|H\right\|_\infty\,.
\end{equation}

Let
\begin{equation}
    H = \sum_{\Lambda\subseteq[n]}H_\Lambda\,,\qquad H_\Lambda\in\mathcal{O}_\Lambda\,.
\end{equation}
We have for any $x\in[n]$
\begin{equation}
    \partial_x H \le 2\left\|\sum_{\Lambda\ni x}H_\Lambda\right\|_\infty \le 2\sum_{\Lambda\ni x}\left\|H_\Lambda\right\|_\infty \le 2\sum_{\Lambda\ni x}c_{|\Lambda|}\left\|H_\Lambda\right\|_\infty\,,
\end{equation}
therefore
\begin{equation}
    \|H\|_L = \max_{x\in[n]}\partial_x H \le 2\max_{x\in[n]}\sum_{\Lambda\ni x}c_{|\Lambda|}\left\|H_\Lambda\right\|_\infty
\end{equation}
and
\begin{equation}
    \|H\|_L \le \|H\|_{\mathrm{loc}}\,.
\end{equation}

The inequality
\begin{equation}
    \frac{\|\cdot\|_1}{2\,c_n} \le \|\cdot\|_{W_1\mathrm{loc}} \le \|\cdot\|_{W_1}
\end{equation}
follows by duality.
\end{proof}

\item The local quantum $W_1$ norm can be upper bounded by the trace norm of the partial traces:
\begin{prop}\label{prop:W1locUB}
    For any $\Delta\in\mathcal{O}_{[n]}^T$ we have
    \begin{equation}\label{eq:DeltaR}
        \|\Delta\|_{W_1\mathrm{loc}} \le \sum_{x\in[n]}\max_{\Lambda\ni x}\frac{\left\|\mathrm{Tr}_{\Lambda^c}\Delta\right\|_1}{2\left|\Lambda\right|c_{|\Lambda|}} \le n\max_{\Lambda\subseteq[n]}\frac{\left\|\mathrm{Tr}_{\Lambda^c}\Delta\right\|_1}{2\left|\Lambda\right|c_{|\Lambda|}}\,.
    \end{equation}
\end{prop}

\begin{proof}
    The claim follows by setting in \eqref{eq:Wloc}
    \begin{equation}
        a_x = \max_{\Lambda\ni x}\frac{\left\|\mathrm{Tr}_{\Lambda^c}\Delta\right\|_1}{2\left|\Lambda\right|c_{|\Lambda|}} \qquad \forall\,x\in[n]\,.
    \end{equation}
\end{proof}

\item As promised, the quantum $W_1$ distance between any two locally indistinguishable states is suppressed by the penalties:
\begin{cor}
Let $\rho,\,\sigma\in\mathcal{S}_{[n]}$ such that $\rho_\Lambda=\sigma_\Lambda$ for any region $\Lambda\subset[n]$ with size $|\Lambda| < k$.
    Then,
    \begin{equation}
        \left\|\rho - \sigma\right\|_{W_1\mathrm{loc}} \le \frac{n}{k\,c_k}\,.
    \end{equation}
\end{cor}

\begin{proof}
    We have from \autoref{prop:W1locUB}
    \begin{equation}
    \left\|\rho - \sigma\right\|_{W_1\mathrm{loc}} \le n\max_{\Lambda\subseteq[n]}\frac{\left\|\rho_\Lambda - \sigma_\Lambda\right\|_1}{2\left|\Lambda\right|c_{|\Lambda|}} \le n\max_{\Lambda\subseteq[n] : |\Lambda|\ge k}\frac{1}{\left|\Lambda\right|c_{|\Lambda|}} = \frac{n}{k\,c_k}\,.
    \end{equation}
    The claim follows.
\end{proof}

\item As the quantum $W_1$ distance, also the local quantum $W_1$ distance is superadditive with respect to the composition of quantum systems and additive with respect to the tensor product:
\begin{prop}\label{prop:tens}
    For any $\Delta\in\mathcal{O}_{m+n}^T$ we have, for any region $\Lambda$ of size $m$,
    \begin{equation}
        \|\Delta\|_{W_1\mathrm{loc}} \ge \|\mathrm{Tr}_{\Lambda^c}\Delta\|_{W_1\mathrm{loc}} + \|\mathrm{Tr}_{\Lambda}\Delta\|_{W_1\mathrm{loc}},
    \end{equation}
    and for any $\rho,\sigma\in\mathcal{S}_{[m+n]}$ we have
    \begin{equation}
        \|\rho-\sigma\|_{W_1\mathrm{loc}} \ge
        \|\rho_{\Lambda}-\sigma_{\Lambda}\|_{W_1\mathrm{loc}} +
        \|\rho_{\Lambda^c}-\sigma_{\Lambda^c}\|_{W_1\mathrm{loc}}.
    \end{equation}
    Moreover, equality is achieved when $\rho = \rho_1\otimes\rho_2$, $\sigma = \sigma_1\otimes\sigma_2$, $\rho_1, \sigma_1\in\mathcal{S}_{[m]}, \rho_2, \sigma_2\in\mathcal{S}_{[n]}$.
\end{prop}

\begin{proof}
    Without loss of generality, we can consider the case where $\Lambda$ is made by the first $m$ qubits.
    Using \eqref{eq:Wloc}, we can then write
    \begin{equation}
        \|\mathrm{Tr}_{m+1,\dots,m+n}\Delta\|_{W_1\mathrm{loc}} = \min\left\{\sum_{x\in[m]}a_x : \frac{\left\|\mathrm{Tr}_{\Lambda^c}\Delta\right\|_1}{2\,c_{|\Lambda|}} \le \sum_{x\in\Lambda}a_x\;\forall\,\Lambda\subseteq[m]\right\}
    \end{equation}
    and
    \begin{equation}
        \|\mathrm{Tr}_{1,\dots,m}\Delta\|_{W_1\mathrm{loc}} = \min\left\{\sum_{x=m+1}^{m+n}a_x : \frac{\left\|\mathrm{Tr}_{\Lambda^c}\Delta\right\|_1}{2\,c_{|\Lambda|}} \le \sum_{x\in\Lambda}a_x\;\forall\,\Lambda\subseteq[m+1,\dots,m+n]\right\}.
    \end{equation}
    Therefore,
    \begin{align}
        &\|\mathrm{Tr}_{1,\dots,m}\Delta\|_{W_1\mathrm{loc}} + \|\mathrm{Tr}_{m+1,\dots,m+n}\Delta\|_{W_1\mathrm{loc}}\nonumber\\
        &= \min\left\{\sum_{x\in[m+n]}a_x : \frac{\left\|\mathrm{Tr}_{\Lambda^c}\Delta\right\|_1}{2\,c_{|\Lambda|}} \le \sum_{x\in\Lambda}a_x\;\forall\,\Lambda\subseteq[n]\ \mathrm{or}\ \Lambda\subseteq[m+1,\dots,m+n]\right\}\nonumber\\
        &\le \min\left\{\sum_{x\in[m+n]}a_x : \frac{\left\|\mathrm{Tr}_{\Lambda^c}\Delta\right\|_1}{2\,c_{|\Lambda|}} \le \sum_{x\in\Lambda}a_x\;\forall\,\Lambda\subseteq[m+n]\right\}\nonumber\\
        &\le \|\Delta\|_{W_1\mathrm{loc}},
    \end{align}
    where the first inequality is a result of the first linear program being the same as in \eqref{eq:Wloc}, with less constraints, which proves the inequality part of the claim.

    For the equality case, we just showed
    \begin{equation}
        \|\rho_1\otimes\rho_2-\sigma_1\otimes\sigma_2\|_{W_1\mathrm{loc}} \ge
        \|\rho_1-\sigma_1\|_{W_1\mathrm{loc}} +
        \|\rho_2-\sigma_2\|_{W_1\mathrm{loc}}\,.
    \end{equation}
    We use \autoref{tensorlem} to prove the other inequality as such
    \begin{align}
        \|\rho_1\otimes\rho_2-\sigma_1\otimes\sigma_2\|_{W_1\mathrm{loc}} &\le \|(\rho_1-\sigma_1)\otimes\rho_2\|_{W_1\mathrm{loc}} + \|\sigma_1\otimes(\rho_2-\sigma_2)\|_{W_1\mathrm{loc}}\nonumber\\
        &\le \|\rho_1-\sigma_1\|_{W_1\mathrm{loc}} +
        \|\rho_2-\sigma_2\|_{W_1\mathrm{loc}}\,,
    \end{align}
    which concludes the proof.
\end{proof}

\item In particular, the local quantum $W_1$ distance recovers the Hamming distance for the states of the computational basis:
\begin{cor}
    For any $x,\,y\in\{0,1\}^n$,
    \begin{equation}
    \left\||x\rangle\langle x| - |y\rangle\langle y|\right\|_{W_1\mathrm{loc}} = h(x,y) = \left|\left\{i\in[n] : x_i\neq y_i\right\}\right|\,.
    \end{equation}
\end{cor}

\begin{proof}
    We have from \autoref{prop:tens}
    \begin{align}
       \left\||x\rangle\langle x| - |y\rangle\langle y|\right\|_{W_1\mathrm{loc}} &= \sum_{i\in[n]} \left\||x_i\rangle\langle x_i| - |y_i\rangle\langle y_i|\right\|_{W_1\mathrm{loc}} = \frac{1}{2}\sum_{i\in[n]} \left\||x_i\rangle\langle x_i| - |y_i\rangle\langle y_i|\right\|_1\nonumber\\
       &= h(x,y)\,.
    \end{align}
\end{proof}

\item The local quantum $W_1$ norm is contracting with respect to the action of single-qubit quantum channels:
\begin{prop}
    For any $\Delta\in\mathcal{O}_{n}^T$, $\Phi$ a quantum channel acting on a single qubit, we have
    \begin{equation}
        \|\Phi(\Delta)\|_{W_1\mathrm{loc}} \le \|\Delta\|_{W_1\mathrm{loc}}.
    \end{equation}
\end{prop}

\begin{proof}
    Using \eqref{eq:Wloc}, we can write
    \begin{equation}
        \|\Phi(\Delta)\|_{W_1\mathrm{loc}} = \min\left\{\sum_{x\in[n]}a_x : \frac{\left\|\mathrm{Tr}_{\Lambda^c}\Phi(\Delta)\right\|_1}{2\,c_{|\Lambda|}} \le \sum_{x\in\Lambda}a_x\;\forall\,\Lambda\subseteq[n]\right\}.
    \end{equation}
    Without loss of generality, $\Phi$ acts on qubit $i$, and there are two cases:
    \begin{enumerate}
        \item If $i\in\Lambda^c$, $\left\|\mathrm{Tr}_{\Lambda^c}\Phi(\Delta)\right\|_1 = \left\|\mathrm{Tr}_{\Lambda^c}\Delta\right\|_1$.
        \item Otherwise, $\left\|\mathrm{Tr}_{\Lambda^c}\Phi(\Delta)\right\|_1 = \left\|\Phi(\mathrm{Tr}_{\Lambda^c}\Delta)\right\|_1 \le \left\|\mathrm{Tr}_{\Lambda^c}\Delta\right\|_1$, since $\Phi$ is a trace-preserving operation.
    \end{enumerate}
    Therefore,
    \begin{equation}
        \|\Phi(\Delta)\|_{W_1\mathrm{loc}} \le \min\left\{\sum_{x\in[n]}a_x : \frac{\left\|\mathrm{Tr}_{\Lambda^c}\Delta\right\|_1}{2\,c_{|\Lambda|}} \le \sum_{x\in\Lambda}a_x\;\forall\,\Lambda\subseteq[n]\right\} = \|\Delta\|_{W_1\mathrm{loc}}.
    \end{equation}
\end{proof}
\end{itemize}

\section{Classical shadows}\label{sec:ClassicalShadows}

The \emph{classical shadow} is a notion introduced to formulate tomographic protocols for extracting information from unknown quantum states with very few measurements (see \autoref{app:shadows} for further details and related works). Let $\mU$ be an ensemble of unitary operators over a $n$-qubit Hilbert space (i.e. any $U\in\mU$ has a statistical weight attached). 
\begin{defn}
$\mU$ is said to be {\bf tomographically complete} if for each $\rho\not = \sigma$ there are $U\in \mU$ and $\ket b$ element of the computational basis $\{\ket b \, :\, b\in\{0,1\}^n\}$ such that:
\beq
\langle b| U\rho U^\dagger |b\rangle \not = \langle b| U\sigma U^\dagger |b\rangle.
\eeq
\end{defn}

\noindent
This definition requires that if two states are different, there is always an evolution $U\in\mU$ such that the evolved states can be distinguished by a measurement in the computational basis.  

Given a tomographically complete ensemble $\mU$ and an unknown state $\rho$ acting on $(\bC^2)^{\otimes n}$, consider the following elementary protocol \cite{huang2020predicting}: 
\\
\\
1. Sample $U\in\mU$ and evolve the state $\rho$;\\
2. Measure the state $U\rho U^\dagger$ in the computational basis.\\
3. Obtained the outcome $\hat b\in\{0,1\}^n$, apply the inverse evolution to ${|\hat b\rangle}$ obtaining $U^\dagger |\hat{b}\rangle\langle{ \hat b}| U$ which can be saved as classical information.
\\
\\
Repeating the three steps above, one can consider the density matrix given by the expectation over $\mU$ and over the possible outcomes: 
\beq
\bE_{U, \hat b} [U^\dagger |\hat{b}\rangle\langle{ \hat b}| U    ]=\bE_{U\sim \mU} \left[\sum_{b\in\{0,1\}^n} \langle b|U\rho U^\dagger|b\rangle \,\,U^\dagger \ket{b}\bra{ b} U\right]=: \mM(\rho),
\eeq
The linear map $\mM$ is completely positive and trace preserving, therefore it is a quantum channel.
Moreover, the requirement of tomographic completeness implies that $\mathcal{M}$ is invertible, though the inverse is in general not completely positive. The {\bf classical shadow} of $\rho$ is defined by:
\beq
\hat \rho:=\mM^{-1}\left(U^{\dagger} |\hat b\rangle\langle\hat b| U\right),
\eeq
it depends on the choice of $\mU$ and it is obtained by a single measurement on $\rho$. By construction, $\bE_{U,\hat b}[\hat \rho]=\rho$. However, $\mM^{-1}$ is not a quantum channel in general, so the classical shadow is computed classically and stored as classical information, as it may not be a quantum state.

Classical shadows can be used to predict expectation values of given observables $O_1,...,O_M$ on the unknown state $\rho$. In fact, the classical shadows define the following random variables:
\beq
\hat o_i:=\tr(O_i \hat \rho)\qquad \forall i\in[M],
\eeq
with the nice property:
\beq
\bE_{U,\hat b}[\hat o_i]=\tr(O_i\rho)\qquad \forall i\in[M].
\eeq

\begin{lem}[\!\!\cite{huang2020predicting}]\label{Var}
Let $O$ be an observable of the $n$-qubit system. The fluctuations of the random variable $\hat o=\tr(O\hat \rho)$ around $\langle O\rangle_{\rho}$ are described by the variance:
\beq
Var[\hat o]= \bE_{U, \hat b}\left[\left(\hat o -\bE_{U,\hat b}[\hat o]\right)^2\right]\leq \left\| O-\frac{\mathrm{Tr}[O]}{2^n}\bI\right\|^2_{sh},
\eeq
where the shadow norm is defined by:
\begin{align}
\| O\|_{sh} &:= \max_{\sigma}\left( \bE_{U\sim \mU}\left[\sum_b  \langle b|U\sigma U^\dagger|b\rangle \,\,{\langle b|U\mM^{-1}(O) U^\dagger \ket{b}}^2\right]\right)^{\frac{1}{2}}\nonumber\\
&= \left\|\bE_{U\sim \mU}\sum_b {\langle b|U\mM^{-1}(O) U^\dagger \ket{b}}^2 \;  U^\dagger|b\rangle\langle b|U\right\|_\infty^\frac{1}{2}.
\end{align}

\end{lem}

\begin{comment}
\begin{prop}[Useful upper bounds for the shadow norm]
For any $O$,
\begin{equation}
    \|O\|_{sh} \le \sqrt{\bE_{U\sim \mU} \max_b{\langle b|U\mM^{-1}(O) U^\dagger \ket{b}}^2} \le \max_{U,\,b}\left|\langle b|U\mM^{-1}(O) U^\dagger \ket{b}\right| \le \left\|\mM^{-1}(O)\right\|_\infty\,.
\end{equation}
\end{prop}
\end{comment}

For any observable $O$, a classical shadow predicts $\tr(O\rho)$ in expectation and the characterization of the variance above allows to boost the convergence providing a good approximation with few measurement processes. Consider $K$ collections of $N$ classical shadows and take the median of the $K$ empirical means over $N$ as an estimator: 
\beq
\hat o(N, K):= \mbox{median}\{\hat o^{(1)}(N,1),...,\hat o^{(K)}(N,1)\},
\eeq
$$
\hat o^{(k)}(N,1)=\frac{1}{N}\sum_{j=N(k-1)+1}^{Nk} \tr(O\hat \rho_j)\qquad k\in[K]
$$

\begin{thm}[\!\!\cite{huang2020predicting}]\label{th.conv}
Given an ensemble $\mU$ of unitaries and observables $O_1,...,O_M$ on an $n$-qubit Hilbert space, let $\epsilon,\delta\in[0,1]$ and set the values:
$$K=2\log\left(\frac{2M}{\delta}\right),$$
$$N=\frac{34}{\epsilon^2}\max_i \left\| O_i-\frac{\mathrm{Tr}(O_i)}{2^n}\bI\right\|^2_{sh}.$$
 Then:
 $$|\hat o_i(N,K)-\mathrm{Tr}(O_i\rho)|\leq \epsilon \qquad\forall i\in[M]$$
with probability at most $1-\delta$.
\end{thm}

The proof of \autoref{th.conv} is based on standard properties of the estimator \emph{median of means}. Since each classical shadow results from a single measurement on $\rho$, the total number of measurements required to estimate the expectation values $\tr(O_i\rho)$ up to error $\epsilon$ is:
\beq
N= O\left(\frac{\log M}{\epsilon^2}\max_i \left\| O_i-\frac{\mathrm{Tr}(O_i)}{2^n}\bI\right\|^2_{sh}\right).
\eeq
The sample complexity is logarithmic in the number of observables we consider and does not depend on the number of qubits. However, there is a dependence on the chosen ensemble $\mU$ via the shadow norm.

\begin{defn}
A {\bf Pauli measurement primitive} is a tomographically complete ensemble $\mathcal U$ such that any $U\in\mathcal U$ is a tensor product $U=U_1\otimes\cdots \otimes U_n$ of randomly selected 1-qubit Clifford gates $U_i\in\mbox{Cl}(2)$. Equivalently, a random Pauli matrix is measured on each qubit.
\end{defn}

Let us focus on the case where classical shadows are constructed applying a Pauli measurement primitive and let us consider the case of $k$-local observables, i.e. $O$ is given by an elementary tensor product supported on $k$ qubits like $O=O_1\otimes\cdots\otimes O_k\otimes \bI^{n-k}$ for instance. According to \cite[Proposition S2 and Lemma S3]{huang2020predicting}, we can state the following result.

\begin{prop}\label{Pauli_shadows}
Let $\rho$ be a $n$-qubit quantum states. The classical shadow of $\rho$ constructed out from a Pauli measurement primitive is:
\beq\label{CS_product}
\hat\rho=\bigotimes_{i=1}^n (3U_i^\dagger \ket{\hat b_i}\bra{\hat b_i}U_i-\bI)\quad\mbox{where}\quad  \hat b_i\in\{0,1\}\quad \forall i=1,...,n.
\eeq
Moreover, let $O$ be a $k$-local observable, then:
\beq
\|O\|_{sh}^2=3^k.
\eeq

\end{prop}

The following \autoref{conv.mean} provides a criterion for estimating the expectation values of local observables using the empirical mean of collected classical shadows. The proof is essentially a consequence of the Bernstein's concentration inequality, a similar result is proved in  \cite{PRXQuantum.3.020365}.

\begin{thm}\label{conv.mean}
Let $O_1,...,O_M$ $k$-local observables, let $\hat\rho_1,...,\hat\rho_N$ be classical shadows of the unknown $n$-qubit state $\rho$ constructed out from a Pauli measurement primitive, and let $\hat\rho:=(1/N)\sum_{i=1}^N\hat\rho_i$ be their empirical mean. Then, for any $\epsilon, \delta>0$, if
\beq
N\geq 3^{k+1}\frac{\log\left(\frac{2M}{\delta}\right)}{\epsilon^2}
\eeq
we have
\beq
|\emph{\tr}(\hat \rho O_m)-\emph{\tr}(\rho O_m)|\leq \epsilon \quad \forall m=1,...,M
\eeq
with probability at least $1-\delta$.

\end{thm}

\begin{proof}
The claim is a consequence of the following inequality: 
\begin{equation}\label{eq:Bern}
    \mathbb P[\max_m|\tr(\hat O_m\hat\rho)-\tr(O_m\rho)|\geq \epsilon]\leq 2 M\exp\left[-\frac{\epsilon^2 N}{3^{k+1}}\right],
\end{equation}
that we can prove applying the Bernstein's concentration inequality: Let $X_1,...,X_N$ be independent random variables such that $\mathbb E[X_i]=0$ and $|X_i|\leq R$ almost surely for all $i=1,...N$. Then, for $\epsilon>0$:
\begin{equation}\label{Bernstein}
    \mathbb P\left[\left|\frac{1}{T}\sum_{i=1}^N X_i\right|\geq\epsilon\right]\leq 2\exp\left(-\frac{\epsilon^2N^2/2}{\sigma^2+RN\epsilon}\right),
\end{equation}
where $\sigma^2=\sum_{i=1}^N \mathbb E[(X_i)^2]$. Given a $k$-local observable $O$, let us define the random variables as $X_i:=|\tr(O\hat\rho_i)-\tr(O\rho)|$ that are independent and centered by construction of the classical shadows. Let $\Lambda\subseteq [n]$, with $|\Lambda| =k$, be the region on which $O$ acts non-trivially. By the H\"older's inequlaity, $\tr(O_\Lambda\rho_\Lambda)\leq \|O_\Lambda\|_\infty \,\,\|\rho_\Lambda\|_1$, we have:
$$|X_i|=|\tr(O_\Lambda\hat\rho_{i\Lambda})-\tr(O_\Lambda \rho_\Lambda)|\leq \|O_\Lambda\|_\infty(\|\rho_\Lambda\|_1+\|\hat\rho_{i\Lambda}\|_1)\leq 1+3^k,$$
where we used the factorized form \eqref{CS_product} of $\hat\rho_i$ and the fact that $\|O_\Lambda\|_\infty\leq 1$. In view of \autoref{Var} and \autoref{Pauli_shadows}, we have $\mathbb E\left[X_i^2\right]\leq\|O\|_{sh}=3^k$, then $\sigma^2\leq N 3^k$. Now, we apply \eqref{Bernstein} obtaining:
\begin{equation}
    \mathbb P\left[|\tr(O\hat\rho)-\tr(O\rho)|\geq\epsilon  \right]\leq 2\exp\left[-\frac{\epsilon^2N^2/2}{N3^k+(1+3^k)N\epsilon}\right]\leq 2\exp\left[-\frac{\epsilon^2N}{3^{k+1}}\right]. 
\end{equation}
The argument above applies for any $O_m$, then:
\begin{align}
\mathbb P\left[\max_m|\tr(O_m\hat\rho)-\tr(O_m\rho)|\geq\epsilon\right] &\leq \sum_{m=1}^M     \mathbb P\left[\max_m|\tr(O_m\hat\rho)-\tr(O_m\rho)|\geq\epsilon\right]\nonumber\\
&\leq 2\, M\exp\left[-\frac{\epsilon^2N}{3^{k+1}}\right],
\end{align}
obtaining \eqref{eq:Bern}. Therefore, $N\geq 3^{k+1}\frac{\log\left(\frac{2M}{\delta}\right)}{\epsilon^2}$ implies
\begin{equation}
\mathbb P\left[\max_m|\tr(O_m\hat\rho)-\tr(O_m\rho)|\geq\epsilon\right]\leq \delta\,,
\end{equation}
that implies in turn that $\mathbb P\left[|\tr(O_m\hat\rho)-\tr(O_m\rho)|\leq\epsilon\right]\geq 1-\delta$ for any $m=1,...,M$ that is the claim.

\end{proof}

$n$-qubit Pauli operators are tensor products of $n$ Pauli matrices (identity included), and the locality or Hamming weight $|P|$ of the Pauli operator $P$ is the number of factors different from the identity, \emph{i.e.}, the number of qubits on which $P$ acts nontrivially.
Classical shadows are not the best protocol to estimate the expectation values of Pauli operators if their locality is high.
Indeed, allowing Bell measurements on two copies of the state, the shadow protocol can be improved obtaining a sample complexity which does not depend on the locality degree of the considered Pauli operators \cite{Huang_2021}.  Let $\{\ket{\Psi^+}, \ket{\Psi^-}, \ket{\Phi^+}, \ket{\Phi^-}\}$ be the Bell basis:
\beq
\ket{\Psi^\pm}=\frac{1}{\sqrt 2}(\ket {00}\pm\ket{11})\qquad \ket{\Phi^\pm}=\frac{1}{\sqrt 2}(\ket{01}\pm\ket{10}),
\eeq
assume to perform a measurement in this basis on each qubit pair of the $n$ qubit pairs in the state $\rho\otimes\rho$. After the Bell measurements on $N_1$ copies of $\rho\otimes\rho$, one obtains a $(2nN_1)$-bit string from which the value $|\tr(P\rho)|$ can be estimated for any Pauli operator $P$. Then, with additional $N_2$ measurements, on can estimate the sign of $\tr(P\rho)$. Any Bell state is an eigenvector of $\sigma\otimes\sigma$ with eigenvalue $\pm 1$ and $\sigma\in\{I,X,Y,Z\}$. Consider the Pauli operator $P=\sigma_1\otimes\cdots\otimes\sigma_n$, then:
\beq
|\tr(P\rho)|^2=\tr[(P\otimes P)(\rho\otimes\rho)]=\mathbb E\left[\prod_{k=1}^n \tr[(\sigma_k\otimes\sigma_k)S_k]\right],
\eeq
where $S_k$ is an eigenprojector of $\sigma_k\otimes\sigma_k$. The average is taken over the distribution of the outcomes of a Bell measurement on any qubit pair in $\rho\otimes\rho$. Let $\{S_k^{(t)}\}$ be the collection of obtained outcomes from the repeated Bell measurements, then we can get the empirical mean as an estimation of $|\tr(P\rho)|^2$:
\beq
\hat a (P)=\frac{1}{N_1}\sum_{t=1}^{N_1}\prod_{k=1}^n \tr[(\sigma_k\otimes\sigma_k)S_k^{(t)}]\,.
\eeq
More precisely, the following proposition is proven in \cite{Huang_2021}:

\begin{prop}[\!\!{\cite{Huang_2021}}]
Given $N_1=\Theta(\log(1/\delta)/\epsilon^4)$ copies of $\rho\otimes\rho$, the following is true for any Pauli operator $P$:
\beq
|\sqrt{\max(\hat a(P),0)}-|\mathrm{Tr}(P\rho)||<\epsilon,
\eeq
with probability $1-\delta$.
\end{prop}
If $|\tr(P\rho)|$ is large enough, then considering $N_2$ copies of $\rho$ allows to estimate the sign by measuring $P$ on each of the $N_2$ copies and taking the majority voting of the obtained $1$s and $-1$s. Assuming $|\tr(P\rho)|$ is large enough, then the majority voting is close to the correct answer, and if $N_2$ is also large enough such that $\rho^{\otimes N_2}$ is not highly perturbed by the measurement, then it can be used to decide the sign for many different Pauli operators. Formally:

\begin{prop}[\!\!{\cite{Huang_2021}}]
For any $\delta,\epsilon,M>0$, let $N_2=\Theta(\log(M/\delta)\epsilon^2)$. For any $M$ Pauli operators $P_1,...,P_M$ with $|\mathrm{Tr}(P_i\rho)|>\epsilon$ for all $i$, measuring $\rho^{\otimes N_2}$, $\mathrm{sign}(\mathrm{Tr}(P_i\rho))$ can be obtained with probability $1-\delta$ for any $i=1,...,M$.    
\end{prop}

The total number of copies of $\rho$ required by the tomographic procedure above to estimate the expectations of $M$ local Pauli operators is $2N_1+N_2$. More precisely, as a consequence of the propositions above, we have the next lemma.

\begin{lem}[\!\!{\cite[Theorem 2]{Huang_2021}}]\label{Bell_lem}
    Given any $P_1, ..., P_M$ Pauli operators and a state $\rho$, there is a procedure that produces $\hat{p}_1, ..., \hat{p}_M$ with
    \begin{equation}
        |\hat{p}_i-\mathrm{Tr}(P_i\rho)| \le \epsilon, \forall i \in [M]
    \end{equation}
    with probability at least $1-\delta$ using $O(\frac{\log(M/\delta)}{\epsilon^4})$ copies of $\rho$.
\end{lem}

Remarkably, the sample complexity of the protocol does not depend on the locality degree (i.e. the number of qubits on which the action is nontrivial) of the considered Pauli operators.

\section{Convergence of the classical shadow in the local quantum \texorpdfstring{$W_1$}{W1} distance}\label{sec:conv}

Let us study the convergence of the empirical mean of the classical shadows to the original state with respect to the local quantum $W_1$ distance:

\begin{thm}\label{thm_CS_2}
Let $\rho$ be an unknown quantum state of $n$ qubits, let $\hat\rho_1,...,\hat\rho_N$ be classical shadows of $\rho$ constructed out from a Pauli measurement primitive, and let $\hat\rho:=\frac{1}{N}\sum_{i=1}^N \hat\rho_i$ be their empirical mean.
We set the coefficients of the local norm (\autoref{local norm}) to $c_l=\frac{c^{l-1}}{l}$, $l=1,...,n$ with $c>\sqrt{10}$, Then, for any $0<\delta<1$ the normalized local $W_1$ distance between $\rho$ and $\hat\rho$ can be bounded as follows:
\beq
\frac{1}{n}\| \hat\rho-\rho\|_{W_1 \mathrm{loc}}\,\leq w,
\eeq
with probability at least $1-\delta$, by a number of classical shadows that scales as:
\beq
N= O\left(\frac{\log\frac{1}{\delta}+\log\frac{1}{w}\log n}{w^2}\right).
\eeq
\end{thm}

\begin{proof}
For a given $k< n$, let us assume we need to estimate the expectations of all the $|P|$-local Pauli operators on $n$ qubits for any $|P|\leq k$, up to an error $\epsilon_P:=3^{|P|/2}(3/c)^k$, using the empirical mean $\hat\rho$. By \autoref{conv.mean}, we need a number $N$ of classical shadows satisfying:
\beq\label{eq:M_k}
N\geq 3 \,\left(\frac{c}{3}\right)^{2k}\log\frac{2 M_k}{\delta}\qquad \qquad M_k=\sum_{i=1}^k \binom{n}{i}3^i \leq \frac{(3n)^k}{(k-1)!},
\eeq
where $M_k$ is the total number of $|P|$-local Pauli operators with $|P|\leq k$. The quantum state $\rho$ and its estimator $\hat \rho$ can be decomposed onto the Pauli basis $\{P\}$:
\beq\label{rho}
\rho=\frac{1}{2^n}\sum_P \langle P\rangle P \quad\mbox{where}\quad \langle P\rangle=\tr(P\rho)\,, \qquad \hat\rho = \frac{1}{2^n}\sum_{|P|\leq k} \hat P P\qquad k\leq n\,,
\eeq
where $\hat{P} = \mathrm{Tr}\left[\hat{\rho}\,P\right]$ is the estimated expectation value of $P$ computed with the empirical mean $\hat{\rho}$ of the classical shadows.
Let us consider the difference of the marginals on a region $\Lambda\subseteq[n]$:
\beq
\hat \rho_\Lambda -\rho_\Lambda=\frac{1}{2^{|\Lambda|}}\sum_{P\in\mathcal P_\Lambda} (\hat P -\langle P\rangle) P,
\eeq
where $\mathcal P_\Lambda$ is the set of local Pauli observables defined on $\Lambda$. Now, we need to bound the trace norm of $\hat \rho_\Lambda -\rho_\Lambda$ and apply \autoref{prop:W1locUB}. Let us recall that, given a $d\times d$ complex matrix $A$, we have $\| A\|_1\leq \sqrt d \| A\|_2$ where $\|\,\,\,\|_2$ is the Hilbert-Schmidt norm.
Therefore:
\beq\label{tracenorm}
\| \hat\rho_\Lambda-\rho_\Lambda\|_1\leq 2^{\frac{|\Lambda|}{2}} \| \hat\rho_\Lambda -\rho_\Lambda\|_2\leq 2^{\frac{|\Lambda|}{2}}\sqrt{\frac{1}{2^{2|\Lambda|}}\sum_{P\in\mathcal P_\Lambda} (\hat P-\langle P\rangle)^2 2^{|\Lambda|}   }=\sqrt{\sum_{P\in\mathcal P_\Lambda} (\hat P-\langle P\rangle)^2    }.
\eeq

\noindent For $|\Lambda|\leq k$, we have $(\hat P-\langle P\rangle)^2\leq\epsilon_P^2$. Thus, applying \eqref{tracenorm}:
\beq
\|\hat \rho_\Lambda-\rho_\Lambda\|_1\leq\sqrt{\sum_{P\in\mathcal P_\Lambda} \epsilon_P^2}= \sqrt{\sum_{i=1}^{|\Lambda|}\binom{|\Lambda|}{i} 9^i \,\left(\frac{3}{c}\right)^{2k}}=\frac{3^k}{c^k} \sqrt{(10^{|\Lambda|}-1)}\leq \sqrt{10^{|\Lambda|}}\,{\frac{3^k}{c^k}},
\eeq
where we used the standard identity $\sum_{i=0}^N\binom{N}{i} x^i=(1+x)^N$ with $N=|\Lambda|$ and $x=9$.
In the case $|\Lambda|>k$, we can bound the trace distance as follows:
\beq
\| \hat\rho_\Lambda-\rho_\Lambda\|_1\leq \|\hat\rho_\Lambda\|_1+1\leq 3^{|\Lambda|}+1,
\eeq
where we have used the triangle inequality and the form of $\hat\rho_\Lambda$ given by \eqref{CS_product}. Therefore, within the choice $c_l=\frac{c^{l-1}}{l}$, for $|\Lambda|\leq k$:
\beq\label{inq_1}
\frac{\|\hat\rho_\Lambda-\rho_\Lambda\|_1}{2|\Lambda|c_{|\Lambda|}}\leq 
\frac{\sqrt{10}^{|\Lambda|}}{2c^{|\Lambda|-1}}{\frac{3^k}{c^k}}\leq \frac{\sqrt{10}}{2}\frac{3^k}{c^k}.
\eeq
For $|\Lambda|\geq k+1$, we have:
\beq\label{inq_2}
\frac{\|\hat\rho_\Lambda-\rho_\Lambda\|_1}{2|\Lambda|c_{|\Lambda|}}\leq\frac{3^{|\Lambda|}+1}{2c^{|\Lambda|-1}}\leq\frac{2}{3} \frac{3^{|\Lambda|}}{ c^{|\Lambda|-1}}\leq 2\, \frac{3^{k}}{c^k}.
\eeq
In \eqref{inq_1} and \eqref{inq_2}, we have used the fact that the terms $\frac{\sqrt{10}^{|\Lambda|}}{c^{|\Lambda|-1}}$ and $\frac{3^{|\Lambda|}}{c^{|\Lambda|-1}}$ are decreasing in $|\Lambda|$ and achieve the maximum for $|\Lambda|=1$ and $|\Lambda|=k+1$ respectively. 
%Within the assumption $c\geq\sqrt{10}$ we have $\frac{\|\hat\rho_\Lambda-\rho_\Lambda\|_1}{2|\Lambda|c_{|\Lambda|}}\leq \sqrt{10}\frac{3^k}{c^k}$.
\noindent
According to \autoref{prop:W1locUB}, we can set: 
\beq\label{eq:w}
\frac{1}{n}\|\hat\rho-\rho\|_{W_1\mathrm{loc}}\leq 2\,{\frac{3^k}{c^k}}=w.
\eeq
In view of \eqref{eq:M_k}, we can estimate the required number of classical shadows to guarantee $\frac{1}{n}\|\hat\rho-\rho\|_{W_1\mathrm{loc}}\leq w$. From \eqref{eq:w}, we obtain that $N\geq 3\, \frac{4}{w^2}\log \left(\frac{2M_k}{\delta}\right) $
and observing that $M_k=O(n^k)$ with $k=O(\log(1/w))$ the claim is proved.

\noindent
%If $c$ is large enough then the right-hand side of \eqref{inq_1} is decreasing in $|\Lambda|$ then the left-hand term is bounded by $\frac{3}{2\sqrt{3^k}}$ otherwise we have:
%\beq
%\frac{\|\hat\rho_\Lambda-\rho_\Lambda\|_1}{2|\Lambda|c_{|\Lambda|}}\leq \frac{\sqrt{10^{k}-1}}{\sqrt{3^k}(c^{k-1}+1)}.
%\eeq

%\noindent
%Using \autoref{prop:W1locUB} and assuming $k\gg 1$, we obtain: 
%\beq
%\frac{1}{n}\|\hat\rho-\rho\|_{W_1 \mathrm{loc}} \leq \max\left(\frac{3}{2\sqrt{3^k}},  \frac{3^{k}}{c^k}     \right)={\frac{3}{2\sqrt{3^k}}}=w,
%\eeq
%if $c$ is large enough. Otherwise:
%\beq
%\frac{1}{n}\|\hat\rho-\rho\|_{W_1 \mathrm{loc}} \leq \max\left(\frac{\sqrt{10^k}}{c^{k}3^k},  \frac{3^{k}}{c^k}  \right)=\frac{3^{k}}{c^k}=w.
%\eeq

%\noindent
%Therefore $\frac{1}{n}\|\hat\rho-\rho\|_{W_1 \mathrm{loc}} \leq w$ is satisfied, with probability $1-\delta$, for:
%\beq
%N=O\left(3^{k+1}{\log\frac{M_k}{\delta}}\right)= O\left(\frac{\log\frac{1}{\delta}+\log\frac{1}{w}\log n}{w^2}\right).
%\eeq
\end{proof}

A second protocol to estimate an unknown $n$-qubit state $\rho$ employs the estimates of the expectation values of all the Pauli operators acting on few qubits.
Let $\hat{P}$ be the estimate of the Pauli operator $P$.
We can then build the following estimate of the state $\rho$:
\beq\label{hat_rho}
\hat\rho:=\frac{1}{2^n}\sum_{P} \hat P P\,.
\eeq
The operator $\hat{\rho}$ of \eqref{hat_rho} may not be positive semidefinite.
However, it satisfies $\mathrm{Tr}\left[\hat{\rho}\,P\right] = \hat{P}$ for any Pauli operator $P$.
%The next theorem establishes the speed of convergence of $\hat \rho$ to $\rho$ in the local $W_1$ distance. 
In \autoref{sec:ClassicalShadows}, we have summarized the tomographic procedure presented in \cite{Huang_2021}, based on Bell measurements, which improves the shadow protocol. Let us consider $\hat\rho$ as defined in \eqref{hat_rho}, where the expectation values of the Pauli operators acting on at most $k$ qubits are estimated by the \emph{Bell procedure}, and the expectation values of the Pauli operators acting on more than $k$ qubits are set to $0$.
We determine in \autoref{Bellconvergence} below the convergence rate of the above estimate to the true state with respect to the local quantum $W_1$ distance.

\begin{thm}\label{Bellconvergence}
    Let $\rho$ be an unknown quantum state of $n$ qubits, $\hat\rho$ be the estimating operator defined in \eqref{hat_rho} constructed out from the Bell procedure accessing $N$ copies of $\rho$.
    Let us set the coefficients of the local norm (\autoref{local norm}) to $c_l = \frac{c^{l-1}}{l}$, $l=1,...,n$ with $c>2$.
    Then, for any $0<\delta<1$ the normalized local $W_1$ distance between $\rho$ and $\hat\rho$ can be bounded as follows:
    \begin{equation}
        \frac{1}{n}\|\hat{\rho} - \rho\|_{W_1 \mathrm{loc}} \le w
    \end{equation}
        with probability $1-\delta$ using a number of copies that scales as:
    \begin{equation}
        N = O\left(\frac{\log\frac{1}{\delta}+\log\frac{1}{w}\log n}{w^4}\right)\,.
    \end{equation}
\end{thm}

\begin{proof}
    The construction of $\hat\rho$ as in \eqref{hat_rho} requires the shadow tomography over all the local Pauli observables up to $k$ qubits, that are
    \beq
    M_k=\sum_{i=1}^k \binom{n}{i} 3^i\leq \sum_{i=1}^k \frac{(3n)^i}{i!} \leq \frac{(3n)^k}{(k-1)!},
    \eeq
for estimating the expectation value of any $P$ with $|P|\leq k$.

    As in the previous proof, we consider the trace norm of the difference of the marginals $\hat\rho_\Lambda-\rho_\Lambda$ on a region $\Lambda\subset[n]$. 
    Since \autoref{Bell_lem} does not depend on the choice of Paulis, we can fix a single $0 < \epsilon < 1$ for all of them. Applying \eqref{tracenorm}, for $|\Lambda|\leq k$, we get with probability $1-\delta$:
    \begin{equation}
        \|\hat{\rho}_\Lambda - \rho_\Lambda\|_1 \le \sqrt{\sum_{P\in\mathcal{P}_\Lambda}\epsilon^2} \le \epsilon\sqrt{\sum_{i=1}^{|\Lambda|}\binom{|\Lambda|}{i}3^i} \le \epsilon2^{|\Lambda|}.
    \end{equation}
    Else if $|\Lambda| \ge k+1$, since we estimated non-local Paulis to be 0, the error for each is at most one:
    \begin{equation}
        \|\hat{\rho}_\Lambda - \rho_\Lambda\|_1 \le \sqrt{\sum_{i=1}^{|\Lambda|}\binom{|\Lambda|}{i}3^i} \le \sqrt{4^{|\Lambda|}-1}.
    \end{equation}
    Applying \autoref{prop:W1locUB}, we have
    \begin{equation}
        \frac{1}{n}\|\hat{\rho} - \rho\|_{W_1 \mathrm{loc}} \le \max_{\Lambda\subseteq[n]}\frac{\|\hat{\rho}_\Lambda - \rho_\Lambda\|_1}{2|\Lambda|c_{|\Lambda|}}.
    \end{equation}
    Choosing $c > 2$, and fixing the coefficients $c_{|\Lambda|} = \frac{c^{|\Lambda|-1}}{|\Lambda|}$, we have, for $|\Lambda| \le k$,
    \begin{equation}
        \frac{\epsilon2^{|\Lambda|-1}}{c^{|\Lambda|-1}} \le \epsilon,
    \end{equation}
    with equality achieved for $|\Lambda| = 1$.
    And for $|\Lambda| \ge k+1$,
    \begin{equation}
        \frac{2^{|\Lambda|-1}}{c^{|\Lambda|-1}} \le \frac{2^k}{c^k},
    \end{equation}
    with equality achieved for $|\Lambda| = k+1$.
    Therefore we have:
    \begin{equation}
        \frac{1}{n}\|\hat{\rho} - \rho\|_{W_1 \mathrm{loc}} \le \max\left(\epsilon, \frac{2^k}{c^k}\right) = w.
    \end{equation}
    It is in our interest to make both arguments of the maximum equal, as it will not change the required number of copies.
    If $\epsilon \le \frac{2^k}{c^k}$, we can increase $\epsilon$ up to $\frac{2^k}{c^k}$ without needing to use extra copies, making both arguments of the $\max$ equal. If it is larger, we can decrease the size of the region $k$ until $\epsilon$ is smaller, and then increase $\epsilon$.
    In the end, we can write:
    \begin{equation}
        \epsilon = \frac{2^k}{c^k}, k = \ceil[\Bigg]{\frac{-\log(\epsilon)}{\log(c/2)}}.
    \end{equation}
    The total number of copies needed, using \autoref{Bell_lem} is:
    \begin{equation}
        N = O\left(\frac{\log(M/\delta)}{w^4}\right) = O\left(\frac{1}{w^4}(\log(1/\delta) + \frac{\log \frac{1}{w}}{\log(c/2)}\log n)\right).
    \end{equation}
\end{proof}

Let us compare the convergence results of \autoref{thm_CS_2} and \autoref{Bellconvergence}.
On the one hand, the number of copies required by the empirical mean of the classical shadows has a better scaling with respect to the local quantum $W_1$ distance compared to the Bell protocol ($O\left(\frac{1}{w^2}\log\frac{1}{w}\right)$ compared to $O\left(\frac{1}{w^4}\log\frac{1}{w}\right)$).
On the other hand, while for the convergence of the Bell-protocol estimate it is enough that the coefficients $c_k$ in the definition of the local quantum norm grow as $\frac{c^k}{k}$ with $c>2$, the convergence of the empirical mean of the classical shadows requires $c>\sqrt{10}$.

\section{Gibbs states}\label{sec:Gibbs}

We have proved in \autoref{sec:conv} that an estimate of a generic state of $n$ qubits that is accurate in the local quantum $W_1$ distance can be obtained by measuring $O(\log n)$ copies of the state.
Measuring $O(\mathrm{polylog}\,n)$ copies of the state $\omega\in\mathcal{S}_{[n]}$ is sufficient to get an estimate that is accurate for the quantum $W_1$ distance of \cite{de2021quantum} if $\omega$ is a Gibbs state of a local Hamiltonian satisfying the transportation cost-inequality
\begin{equation}\label{eq:TC}
\left\|\rho - \omega\right\|_{W_1}^2 \le \frac{n\,C}{2}\,S(\rho\|\omega)\qquad\forall\;\rho\in\mathcal{S}_{[n]}\,,
\end{equation}
which upper bounds the quantum $W_1$ distance between $\omega$ and a generic state $\rho$ with their quantum relative entropy \cite{rouze2023learning,onorati2023efficient}.
Such transportation-cost inequality has been proved for the Gibbs states of local Hamiltonians satisfying suitable forms of decay of correlations \cite{de2022quantum,onorati2023efficient}.

In this section, we connect the two results by proving that when we restrict the quantum $W_1$ distance and the local quantum $W_1$ distance to any family of Gibbs states of Hamiltonians with local quantum norm $O(1)$ satisfying the transportation-cost inequality \eqref{eq:TC}, the two distances become equivalent: 
\begin{prop}
For any Hamiltonian $H\in\mathcal{O}_{[n]}$, let
\begin{equation}
\omega_H = \frac{e^{-H}}{\mathrm{Tr}\,e^{-H}}
\end{equation}
be the associated Gibbs state (the inverse temperature does not appear since it can be reabsorbed in $H$).
Let $\mathcal{F}\subseteq\mathcal{O}_{[n]}$ be a family of Hamiltonians with local quantum norm at most $M$.
Let us assume that for any $H\in\mathcal{F}$, the Gibbs state $\omega_H$ satisfies the transportation-cost inequality \eqref{eq:TC} with a constant $C$ that does not depend on $H$.
Then, for any $H,\,K\in\mathcal{F}$ we have
\begin{equation}\label{eq:loc2}
\frac{\left\|\omega_H-\omega_K\right\|_{W_1\mathrm{loc}}^2}{n^2} \le \frac{\left\|\omega_H-\omega_K\right\|_{W_1}^2}{n^2} \le \frac{M\,C}{2}\,\frac{\left\|\omega_H-\omega_K\right\|_{W_1\mathrm{loc}}}{n}\,.
\end{equation}
\end{prop}

\begin{proof}
The first inequality in \eqref{eq:loc2} follows from \autoref{prop:Lloc}.
We have from \eqref{eq:TC}
\begin{align}
\frac{\left\|\omega_H-\omega_K\right\|_{W_1}^2}{n^2} &\le \frac{C}{4\,n}\left(S(\omega_H\|\omega_K) + S(\omega_K\|\omega_H)\right) = \frac{C}{4\,n}\,\mathrm{Tr}\left[\left(\omega_H-\omega_K\right)\left(K - H\right)\right]\nonumber\\
&\le \frac{C}{4\,n}\left\|\omega_H-\omega_K\right\|_{W_1\mathrm{loc}}\left\|K - H\right\|_{\mathrm{loc}} \le \frac{C\,M}{2\,n}\left\|\omega_H-\omega_K\right\|_{W_1\mathrm{loc}}\,.
\end{align}
The claim follows.
\end{proof}

\section{Quantum Wasserstein Generative Adversarial Networks}\label{sec:GAN}
Quantum Generative Adversarial Networks (QGANs) constitute an algorithm to train a parametric quantum circuit to learn an unknown quantum state \cite{lloyd2018quantum}.
The training takes the form of an adversarial game, where a generator parametric quantum circuit with the goal of generating a state as close as possible to the true state is trained against a discriminator with the goal of discriminating between the true state and the generated state.
In the typical setup, the discriminator trains a parametric observable to maximize the difference between its expectation value on the true state and on the generated state.
The choice of the parametric observable plays a crucial role for the success of the training.
In the original proposal of \cite{lloyd2018quantum}, the observable is constrained to have operator norm at most one, such that if the available set of parametric observables is large enough, the discriminator obtains the trace distance between the true and the generated state.

This choice has later been shown to suffer from the problem of barren plateaus, \emph{i.e.}, the gradient of the cost function decays exponentially with the number of qubits and quickly becomes indistinguishable from zero, thus making the training impossible \cite{mcclean2018barren}.
This problem can be ascribed to the property that any two orthogonal states have trace distance equal to one.
Therefore, if we want to obtain the state $|1\rangle^{\otimes n}$ starting from the state $|0\rangle^{\otimes n}$ and we proceed by flipping the qubits one by one, the trace distance will not notice any progress until the last qubit is flipped.

To solve this problem, Ref. \cite{kiani2022learning} has proposed a quantum Wasserstein GAN (QWGAN) where the discriminator optimizes his cost over observables with quantum Lipschitz constant at most one.
In this case, if the available set of parametric observables is large enough, the discriminator obtains the quantum $W_1$ distance between the true and the generated state.
This choice was inspired both by the predominance of the Wasserstein distance as cost function of the classical GANs \cite{arjovsky2017wasserstein} and by the results of Ref. \cite{cerezo2021cost} proving that local cost functions computed at the output of quantum circuits with logarithmic depth do not suffer from barren plateaus.
Ref. \cite{kiani2022learning} shows that, contrarily to the original QGAN, the QWGAN is capable of learning complex quantum states, such as the $n$-qubit GHZ state.

In practice, the computational complexity of computing the exact Lipschitz constant grows exponentially with the number of qubits.
Therefore, the QWGAN of \cite{kiani2022learning} actually replaces the quantum Lipschitz constant with the upper bound given by the local quantum norm of the present paper with all the coefficients $c_k$ set to one.
Moreover, since the dimension of the vector space of the observables grows exponentially with the number of qubits $n$, the QWGAN of \cite{kiani2022learning} restricts the optimization of the discriminator to the linear combinations of a set of $O(\mathrm{poly}\,n)$ tensor products of Pauli matrices.
If no a priori information on the state to be learnt is available, the most natural choice for such set is made by the tensor product of few Pauli matrices.
With this choice, the constraint on the observable becomes effectively a constraint on its local quantum norm, and the QWGAN will measure the quality of the generated state with respect to the local quantum $W_1$ distance.

We have proved in \autoref{sec:conv} that the classical shadow obtained by measuring $O(\mathrm{poly} n)$ copies of any quantum state constitutes an accurate estimate with respect to the local quantum $W_1$ distance.
Therefore, our results imply that the QWGAN can be equivalently trained using the classical shadow in place of the true state and does not get any advantage in having quantum access to the true state, unless some prior information on the true state motivates the addition of some tensor product of many Pauli matrices to the set of observables available to the discriminator.
Indeed, the successful learning of the $n$-qubit GHZ state by the QWGAN of \cite{kiani2022learning} was based on such an addition.

\section{Conclusions}\label{sec:concl}
We have defined the local quantum $W_1$ distance as a distance that captures the notion of local distinguishability and we have proved that the classical shadow produced by measuring $O(\log n)$ copies of any state of $n$ qubits provides an estimate of the state which is accurate with respect to the local quantum $W_1$ distance. In particular, we have determined the speed of convergence toward the true state of the estimate given by the empirical mean of a collection of classical shadows (\autoref{thm_CS_2}). Moreover, we have considered the tomographic protocol presented in \cite{Huang_2021} that improves the shadow protocol by means of Bell measurements. Also in this case we have determined the speed of convergence of the estimate to the true quantum state in the local quantum $W_1$ distance (\autoref{Bellconvergence}). 

Moreover, we have proved that when restricted to the set of Gibbs states of local Hamiltonians which can be efficiently estimated with respect to the quantum $W_1$ distance, the local quantum $W_1$ distance is equivalent to the quantum $W_1$ distance.
Furthermore, we have applied our results to quantum generative adversarial networks, showing that the QWGAN proposed in \cite{kiani2022learning} can get advantages from having quantum access to the state to be learned only when some prior information on such state is available.

Fundamental questions that are left open are whether the convergence speeds of \autoref{thm_CS_2} and \autoref{Bellconvergence} are optimal, and whether the requirements of such theorems on the scaling of the coefficients $c_k$ in the definition of the local quantum norm can be relaxed.

\section*{Acknowledgements}
GDP was supported by the HPC Italian National Centre for HPC, Big Data and Quantum Computing - Proposal code CN00000013 and by the Italian Extended Partnership PE01 - FAIR Future Artificial Intelligence Research - Proposal code PE00000013 under the MUR National Recovery and Resilience Plan funded by the European Union - NextGenerationEU.
DP was supported by project SERICS (PE00000014) under the MUR National Recovery and Resilience Plan funded by the European Union - NextGenerationEU.
GDP is a member of the ``Gruppo Nazionale per la Fisica Matematica (GNFM)'' of the ``Istituto Nazionale di Alta Matematica ``Francesco Severi'' (INdAM)''.

\appendix

\section{Further approaches to quantum optimal mass transport}\label{app:appr}

Several quantum generalizations of optimal transport distances have been proposed besides the one of Ref. \cite{de2021quantum}.
One line of research by Carlen, Maas, Datta and Rouz\'e \cite{carlen2014analog,carlen2017gradient,carlen2020non,rouze2019concentration,datta2020relating,van2020geometrical,wirth2022dual} defines a quantum Wasserstein distance of order $2$ from a Riemannian metric on the space of quantum states based on a quantum analog of a differential structure.
Exploiting their quantum differential structure, Refs. \cite{rouze2019concentration,carlen2020non,gao2020fisher} also define a quantum generalization of the Lipschitz constant and of the Wasserstein distance of order $1$.
Alternative definitions of quantum Wasserstein distances of order $1$ based on a quantum differential structure are proposed in Refs. \cite{chen2017matricial,ryu2018vector,chen2018matrix,chen2018wasserstein}.
Refs. \cite{agredo2013wasserstein,agredo2016exponential,ikeda2020foundation} propose quantum Wasserstein distances of order $1$ based on a distance between the vectors of the canonical basis.

Another line of research by Golse, Mouhot, Paul and Caglioti \cite{golse2016mean,caglioti2021towards,golse2018quantum,golse2017schrodinger,golse2018wave, caglioti2019quantum,friedland2021quantum, cole2021quantum, duvenhage2021optimal,bistron2022monotonicity,van2022thermodynamic} arose in the context of the study of the semiclassical limit of quantum mechanics and defines a family of quantum Wasserstein distances of order $2$ built on a quantum generalization of couplings.
Such distances have been generalized to von Neumann algebras \cite{duvenhage2020quadratic,duvenhage2021wasserstein,duvenhage2022extending}.

Ref. \cite{de2021quantumAHP} proposes another quantum Wasserstein distance of order $2$ based on couplings, with the property that each quantum coupling is associated to a quantum channel.
The relation between quantum couplings and quantum channels in the framework of von Neumann algebras has been explored in \cite{duvenhage2018balance}.
The problem of defining a quantum Wasserstein distance of order $1$ through quantum couplings has been explored in Ref. \cite{agredo2017quantum}.

The quantum Wasserstein distance between two quantum states can be defined as the classical Wasserstein distance between the probability distributions of the outcomes of an informationally complete measurement performed on the states, which is a measurement whose probability distribution completely determines the state.
This definition has been explored for Gaussian quantum systems with the heterodyne measurement in Refs. \cite{zyczkowski1998monge,zyczkowski2001monge,bengtsson2017geometry}.

\section{Related works on classical shadows}\label{app:shadows}

Tomography with classical shadows, summarized in \autoref{sec:ClassicalShadows}, has been proposed as a restricted version of the shadow tomography protocol originally developed by Aaronson \cite{aaronson2020shadow}. In general, the shadow tomography addresses this problem: given a collection of $m$ observables, how many copies of an $n$-qubit state $\rho$ are necessary and sufficient to estimate their expectation values over $\rho$ up to an error $\epsilon$? A crucial requirement is to avoid considering an exponential number of copies of the unknown state as done in standard quantum state tomography. Using post selected learning \cite{aaronson2007learnability}, shadow tomography achieves a sample complexity $\widetilde O((n\log^4 m)/\epsilon^4)$. However, the original shadow tomography protocol presents an exponential time complexity, in this respect  classical shadows provide a more efficient protocol \cite{huang2020predicting}. Moreover, tomography with classical shadows has been analyzed in presence of noise \cite{Koh2022classicalshadows}, extended to continuous variables quantum systems \cite{becker2023classical}, characterized in terms of Bayesian analysis \cite{Lukens2021bayesian}, and applied in several contexts \cite{PRXQuantum.3.020365, becker2023classical, zhang2021experimental, Hadfield2022measurements, zhao2021fermionic}.

A recent theoretical generalization of classical shadows, called \emph{hybrid shadows} \cite{shivam2023classical}, has been proposed. In this case, given an $n$-qubit state, some of the qubits are measured to store classical shadows and the entangled states of the remaining qubits are stored as quantum data. This technique can be used for providing more accurate estimates of expectations values at the cost of more quantum memory.   

Another recent generalization of classical shadow tomography has been proposed considering unitary ensembles where the probability distribution of the evolution unitaries is invariant under local-basis transformations, such as random unitary circuits and quantum Brownian dynamics \cite{hu2023classical}.

Beyond classical shadows, there are other improvements of the shadow tomography. For instance, Badescu and O’Donnell \cite{badescu2021improved}, improved the sample complexity of shadow tomography to $\widetilde O((n^2\log^2m)/\epsilon^2)$ based on a procedure called \emph{quantum hypothesis selection} which can be viewed as an \emph{agnostic learning} of quantum states \cite{anshu2023survey}.

Shadow tomography with only allowed separable measurements is considered by Chen et al. \cite{chen21exponential}, they proved that $\widetilde\Omega(\min\{m,d\})$ copies of a $d$-dimensional quantum states are necessary for estimating $m$ expectation values. This sample complexity matches to the upper bound $\widetilde O(\min\{m,d\})$ showed in the first proposal of classical shadow tomography \cite{huang2020predicting}.

\section{Auxiliary Lemmas}\label{app:lem}

\begin{lem}\label{tensorlem}
    For any $\Delta_1\in\mathcal{O}_{m}^T,\Delta_2\in\mathcal{O}_{n}^T$, we have
    \begin{equation}
        \|\Delta_1\otimes\Delta_2\|_{W_1\mathrm{loc}} \le
        \|\Delta_1\|_{W_1\mathrm{loc}}\|\Delta_2\|_1.
    \end{equation}
\end{lem}

\begin{proof}
    Using \eqref{eq:Wloc}, we can write
    \begin{equation}
        \|\Delta_1\otimes\Delta_2\|_{W_1\mathrm{loc}} = \min\left\{\sum_{x\in[m+n]}a_x : \frac{\left\|\mathrm{Tr}_{\Lambda^c}\Delta_1\otimes\Delta_2\right\|_1}{2\,c_{|\Lambda|}} \le \sum_{x\in\Lambda}a_x\;\forall\,\Lambda\subseteq[m+n]\right\}.
    \end{equation}
    We can then separate $\Lambda$ in two, one part acting on $\Delta_1$ and another acting on $\Delta_2$: $\Lambda=\Lambda_1\cup\Lambda_2$. We can therefore write
    \begin{equation}
        \|\Delta_1\otimes\Delta_2\|_{W_1\mathrm{loc}} = \min\left\{\sum_{x\in[m+n]}a_x : \frac{\left\|\mathrm{Tr}_{\Lambda_1^c}\Delta_1\right\|_1 \left\|\mathrm{Tr}_{\Lambda_2^c}\Delta_2\right\|_1}{2\,c_{|\Lambda|}} \le \sum_{x\in\Lambda}a_x\;\forall\,\Lambda\subseteq[m+n]\right\}.
    \end{equation}
    Since $\left\|\mathrm{Tr}_{\Lambda_2^c}\Delta_2\right\|_1 \le \|\Delta_2\|_1$, and $c_{|\Lambda_1|} \le c_{|\Lambda|}$, we get
    \begin{equation}
        \|\Delta_1\otimes\Delta_2\|_{W_1\mathrm{loc}} \le
        \min\left\{\sum_{x\in[m+n]}a_x : \frac{\left\|\mathrm{Tr}_{\Lambda_1^c}\Delta_1\right\|_1 }{2\,c_{|\Lambda_1|}}\left\|\Delta_2\right\|_1 \le \sum_{x\in\Lambda}a_x\;\forall\,\Lambda\subseteq[m+n]\right\}.
    \end{equation}
    There is no longer any dependency on $\Lambda_2$, which means $\forall x \ge m+1, a_x = 0$. Therefore we can extract $\|\Delta_2\|_1$,
    \begin{align}
        \|\Delta_1\otimes\Delta_2\|_{W_1\mathrm{loc}} &\le
        \min\left\{\sum_{x\in[m]}a_x : \frac{\left\|\mathrm{Tr}_{\Lambda_1^c}\Delta_1\right\|_1 }{2\,c_{|\Lambda_1|}} \le \sum_{x\in\Lambda_1}a_x\;\forall\,\Lambda_1\subseteq[m]\right\}\left\|\Delta_2\right\|_1\nonumber\\
        &=\|\Delta_1\|_{W_1\mathrm{loc}}\|\Delta_2\|_1.
    \end{align}
\end{proof}

\begin{comment}
\begin{thm}[{\!\!\cite[Theorem 3]{haah2017sample}}]\label{lower bound}
Given $\epsilon,\,\delta\in(0,1)$ and $N,\,d\in\mathbb{N}$, let us assume that there exists a quantum tomographic protocol that for any $d\times d$ density matrix $\rho$, by measuring $N$ independent copies of $\rho$ generates an estimate $\hat{\rho}\in\mathbb{C}^{d\times d}$ such that $\|\hat\rho-\rho\|_1\leq \epsilon$ with probability at least $1-\delta$.
Then:
\begin{equation}
    N\geq C_\delta\frac{d^2}{\epsilon^2}(1-\epsilon^2)\,,
\end{equation}
where $C_\delta$ is a constant depending uniquely on $\delta$.
\end{thm}

\begin{rem}
    The statement of \cite[Theorem 3]{haah2017sample} assumes that $\hat{\rho}$ is a density matrix. However, the proof never uses this assumption and the result holds for any estimate in $\mathbb{C}^{d\times d}$.
\end{rem}
\end{comment}

\bibliography{biblio}
\bibliographystyle{unsrt}
\end{document}